\documentclass{llncs}
\usepackage{amssymb,amsmath}
\usepackage{nicefrac}
\usepackage{url}
\usepackage{color}
\usepackage[american]{babel}
\usepackage{graphicx}
\usepackage{version}
\usepackage{subfigure} % change to sub-caption if you'd rather be fancy
\usepackage[textsize=footnotesize]{todonotes}

\newif\iffull
\fulltrue
\graphicspath{{figures/}}

\usepackage{amsthm}
\newtheorem{observation}{Observation}

\newcommand{\leaveout}[1]{}

\date{}
\title{Visibility Representations of Toroidal and Klein-bottle Graphs}
\author{Therese Biedl
\orcidID{0000-0002-9003-3783}
\thanks{Supported by NSERC.  The author would like to thank Sam Barr
for helpful input.}
}
\institute{David R.~Cheriton School of Computer
Science, University of Waterloo, Waterloo, Ontario N2L 3G1, Canada.
\email{biedl@uwaterloo.ca}
}

\begin{document}

\maketitle
\begin{abstract}
In this paper, we study visibility representations of graphs that are embedded on a torus or a Klein bottle.  Mohar and Rosenstiehl showed that any toroidal graph has a visibility representation on a flat torus bounded by a parallelogram, but left open the question whether one can assume a rectangular flat torus, i.e., a flat torus bounded by a rectangle.   Independently the same question was asked by Tamassia and Tollis.   We answer this question in the positive.    With the same technique, we can also show that any graph embedded on a Klein bottle has a visibility representation on the rectangular flat Klein bottle.
\end{abstract}

%\linenumbers

%%%%%%%%%%%%%%%%%%%%%%%%%%%%%%%%%%%%%%%%%%%%%%%%%%%%%%%%%%%%%%%%%%%%%%%%
\section{Introduction}

Visibility representations are one of the oldest topics studied in
graph drawing.   Introduced as {\em horvert-drawings} by Otten and
Van Wijk in 1978 \cite{OW78}, and independently as {\em S-representations} by
Duchet, Hamidoune, Las Vergnas and Meyniel in 1983 \cite{DHLM83}, 
they consist of assigning disjoint horizontal segments to vertices and
disjoint vertical segments to every edge such that for each edge the
segment 
ends at the two vertex-segments of its endpoints and intersects no 
other vertex-segment.  (Fig.~\ref{fig:K7_I}(d) gives an example.)
Later papers studied exactly which planar graphs
have such visibility representations \cite{RT86,TT86,Wis85} and 
generalized them to the rolling cylinder 
\cite{TT91}, M\"obius band \cite{Dean00}, projective plane \cite{Hut05}
or torus \cite{MR98}.    (There are numerous other generalizations,
e.g. to higher dimensions \cite{BEG+98}, or permitting rectangles for vertices
and horizontal and vertical edges \cite{BDHS97}, or permitting edges to go through a limited
set of vertex-segments \cite{DEG07}.)

The motivation for the current paper is the work by Mohar and Rosenstiehl
\cite{MR98}, who showed that any {\em toroidal graph} (i.e., a graph that can be
drawn on a torus without crossings) has a visibility representation on
the {\em flat torus}, i.e., a parallelogram $Q$ where opposite edges have been
identified.   They explicitly stated as open problem whether the same holds
for a {\em rectangular flat torus}, i.e., where $Q$ must be a rectangle---their 
method cannot be generalized to this case.  (See also Fig.~\ref{fig:M0}.)
The same question was asked independently earlier by Tamassia and Tollis \cite{TT91}.
This paper answers this question in the positive. 

\begin{theorem}
\label{thm:tor}
Let $G$ be a toroidal graph without loops.   Then $G$ has a visibility representation
on the rectangular flat torus.
\end{theorem}

There are quite a few graph drawing results for toroidal graphs;
see Castelli Aleardi et al.~\cite{CDF18} and the references therein 
for increasingly better results for straight-line drawings.
Their approach is to convert the toroidal graph into a
planar graph by deleting edges, then draw this planar graph, 
and then reinsert the edges.   (Other papers \cite{Hut05,MR98} instead use
a reduction approach, where the graph-size is reduced while staying
in the same graph class until some small graph is reached, draw this graph,
and then undo the reduction in the drawing.)   We follow the first approach
(i.e., delete edges to make the graph planar), but face a major challenge
when wanting to reinsert an edge $(v,w)$.   For this, we need
the segments of $v$ and $w$ to be visible across the horizontal boundary
of the fundamental rectangle, and in particular, to share an $x$-coordinate. 
We achieve this by keeping two halves of each removed edge, connecting
corresponding half-edges along paths, and then forcing these paths to
be drawn along columns; the ability to do so may be of independent interest.

\iffalse
Our paper is structured as follows.   After reviewing some background in
Section~\ref{sec:background}, we first explain how to create visibility
representations on the flat cylinder in Section~\ref{sec:cyl}.   The
existence of such drawings was known \cite{Hut05}, but we need to re-prove
the result so that we can argue that we can also make the representations
path-aligned.   In Section~\ref{sec:tor} we turn to toroidal graphs, 
explain our particular method of splitting them to obtain a planar graph 
with a specified source $s$ and sink $t$.   With a suitable choice of
path system, this then 
gives the desired visibility representation, and with a minor tweak
the same approach works for Klein-bottle graphs.
%We briefly discuss some aspects of our drawings and related results in
%Section~\ref{sec:related} before concluding 
We conclude in Section~\ref{sec:concl}.
\fi

\begin{figure}[t]
\includegraphics[width=0.18\linewidth,page=2]{K7.pdf}
\hspace*{\fill}
\includegraphics[width=0.22\linewidth,page=1]{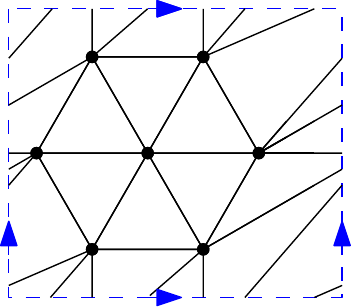}
\hspace*{\fill}
\includegraphics[width=0.18\linewidth,page=2]{Petersen.pdf}
\hspace*{\fill}
\includegraphics[width=0.22\linewidth,page=1]{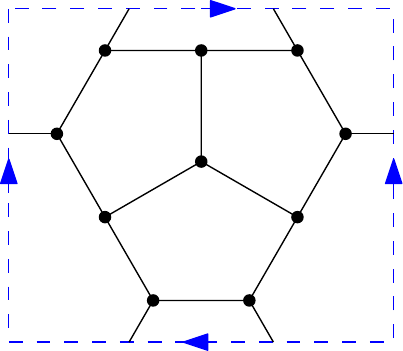}
\caption{The complete graph $K_7$ embedded on the rectangular flat torus and the Petersen-graph embedded on the rectangular flat Klein bottle.}
\label{fig:K7}
\label{fig:Pet}
\end{figure}
%%%%%%%%%%%%%%%%%%%%%%%%%%%%%%%%%%%%%%%%%%%%%%%%%%%%%%%%%%%%%%%%%%%%%%%%
\section{Background}
\label{sec:background}

We assume familiarity with graph theory and planar graphs, see for
example Diestel's book \cite{Die12}.   
\iffalse
Let $G=(V,E)$ be an undirected connected graph 
with $n$ vertices and $m$ edges.   A {\em bipolar orientation} of $G$
is an assignment of directions to the edges that is acyclic and has exactly
one source and one sink.  An {\em $st$-order} is a vertex order $v_1,\dots,v_n$ 
such that orienting all edges from the lower-indexed to the higher-indexed
vertex gives a bipolar orientation.   Vice versa, for any bipolar orientation
enumerating the vertices in topological order gives an $st$-order.   It is
well-known that any 2-connected graph has a bipolar orientation, even if we
fix a-priori which vertices should be the source and sink \cite{LEC67}, and
we can find it in linear time \cite{ET76}.
\fi
Throughout, let $G=(V,E)$ be a connected graph without loops, with $|V|=n$ and $|E|=m$.
A {\em map} $M$ on a surface $\Sigma$ is a 2-connected graph $G$ together with
an embedding of $G$ in $\Sigma$ such that every {\em face} (i.e., connected region of 
$\Sigma\setminus M$) is bounded by a simple cycle.   
%It is well-known that 
Maps correspond naturally to rotation systems on the underlying
graphs, up to homomorphisms among the embeddings \cite{GT87}.   Here
a {\em rotation system} is a set of cyclic permutations $\rho_v$ (for
$v\in V$) where $\rho_v$ corresponds to the clockwise cyclic order in which the edges incident
to $v$ emanate from $v$ in the embedding.      
%The rotation system defines
%the {\em faces} (corresponding to the connected regions of $\Sigma\setminus M$).
For ease of description we often assume that we have a map, though all
algorithmic steps could be performed on the rotation system alone.
%and the {\em dual map} $M^*$ whose vertices are the faces of $M$ with an edge 
%connecting two vertices in $M^*$ if and only if the corresponding faces of $M$
%share an edge.

\iffalse
In this paper we chiefly consider {\em planar graphs} (which come with
a map in the plane) and {\em toroidal graphs} (which come with a map on
the torus); we also briefly study {\em Klein-bottle graphs} (which come
with a map on the Klein bottle).   Throughout we assume that all faces are 
bounded by simple cycles, which for a planar graph implies that it is 2-connected.
A {\em non-contractible cycle} of a map is a cycle of the graph that
forms a non-contractible curve, i.e., a curve that cannot be continuously
deformed to a point.
\fi

We study surfaces that have a {\em flat representation} 
consisting of a {\em fundamental parallelogram} $Q$ in the plane with some 
sides identified.   (We may assume that two sides of $Q$ are horizontal, 
hence $Q$ has a left/right/top/bottom side.)   
A {\em (standing) flat cylinder} is obtained by identifying the left and right side of $Q$
in the same direction (bottom-to-top).  (We usually omit `standing' since we will
not discuss other kinds.)
A {\em flat torus} is obtained from a flat cylinder by identifying the top and
bottom side in the same direction (left-to-right), while
a {\em flat Klein bottle} is obtained from a flat cylinder by identifying the top and
bottom side in opposite direction.
Figs.~\ref{fig:K7}, \ref{fig:M0}, \ref{fig:K44} give some examples.   
A {\em rectangular flat torus} [\emph{rectangular Klein bottle}] is a flat torus [flat Klein bottle]
for which the fundamental parallelogram $Q$ is required to be a rectangle.

Flat representations
carry the local geometry of the plane; in particular when we speak of a
{\em segment} or an {\em $x$-interval} then we specifically permit it to
go across a side of the fundamental parallelogram $Q$.   So for example in a flat cylinder
$Q=[0,w]\times [0,h]$, 
an $x$-interval can have the form $[x',x'']$ for two $x$-coordinates $x'<x''$,
but it can also have the form $[0,x''] \cup [x',w]$ for some $x''<x'$.
A {\em row/column} of $Q$ is a horizontal/vertical line with integer
coordinate that intersects the interior of $Q$.

A {\em visibility representation} of a graph $G$ is a mapping of
vertices into non-overlapping horizontal segments (called {\em vertex-segments})
and of edges of $G$ into non-overlapping 
vertical segments (called {\em edge-segments}) such
that for each edge $(u,v)$, the associated edge-segment has its endpoints
on the vertex-segments corresponding to $u$ and $v$ and it does not intersect 
any other vertex-segment.   

%%%%%%%%%%%%%%%%%%%%%%%%%%%%%%%%%%%%%%%%%%%%%%%%%%%%%%%%%%%%%%%%%%%%%%%%
\section{Creating Visibility Representations}

We first give an outline of our approach.   Quite similar to what was
was done for straight-line drawings of toroidal graphs 
\cite{CDF18}, we remove a set of edges
to convert the given graph into a planar graph.   In contrast to
the earlier work, we keep the edges but split each of them into two
`half-edges' that end at two new vertices $s,t$ (Section~\ref{subse:makePlanar}).
We will later need to re-connect these half-edges, and to this end,
choose a `path-system' that connects each pair of half-edges while
keeping all the paths non-crossing and (after duplicating some edges)
edge-disjoint (Section~\ref{subse:choosePaths}).    Then we create
a visibility representation on the flat cylinder
for which these paths are drawn vertically. To be able to do so we
first must argue that we can find an $st$-order that enumerates
vertices of all paths in order (Section~\ref{subse:stOrder}).  Then we
build the visibility representation (Section~\ref{subse:VR}).   Removing the
segments of $s$ and $t$ and possibly inserting more columns
%, the visibility representation then resides on 
%the flat torus  and the flat Klein bottle, respectively.
gives the desired visibility representation.
Figs.~\ref{fig:K7_I} and \ref{fig:G_st}-\ref{fig:K7VR} illustrate the approach for $K_7$ and
the Petersen-graph. 

\def\myScale{0.68}
\begin{figure}[ht]
\hspace*{\fill}
\begin{tabular}{ccc}
{\includegraphics[scale=\myScale,page=3,trim=0 10 0 0,clip]{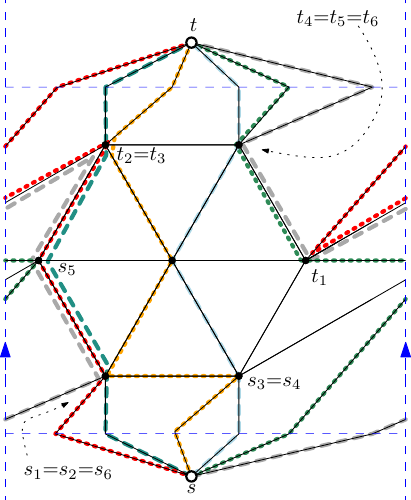}}
& ~~~~~
& {\includegraphics[scale=\myScale,page=4,trim=0 10 0 0,clip]{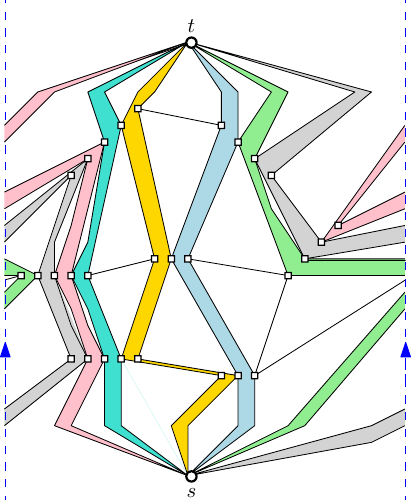}}  \\
(a) $G_{st}$ and the path-system. && (b) Splitting at the paths.\\[2ex]
{\includegraphics[scale=\myScale,page=5,trim=0 10 0 0,clip]{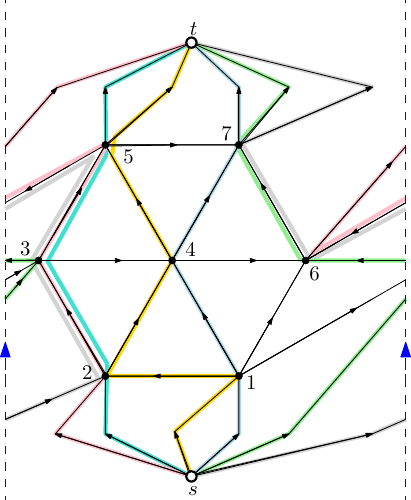}}
& 
&{\includegraphics[scale=0.8,page=6,trim=0 10 0 0,clip]{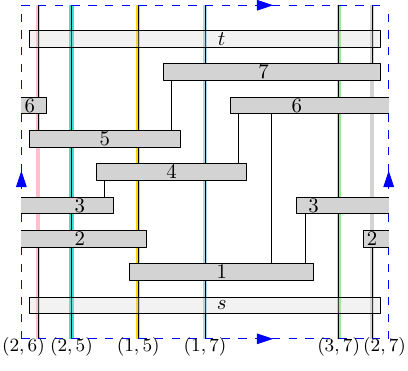}}\\
(c) Orientation and $st$-order. && (d) The final drawing.
\end{tabular}
\hspace*{\fill}
\caption{The construction for the complete graph $K_7$.}
\label{fig:K7_I}
\end{figure}

%----------------------------------------------------------------------
\subsection{Making the Graph Planar}
\label{subse:makePlanar}

In this section we explain how to modify the input graph $G$ to make
it planar.   We assume that $G$ has no loop and comes embedded on a flat realization $Q$
(either a torus or a Klein bottle).
We first modify this embedding to achieve the following:
(1) Every face is bounded by a simple cycle, so the embedding is a map.
(2) No edge crosses the horizontal boundary of $Q$ twice.   
(3) Parallelogram $Q$ is a rectangle.
(4) No vertex lies on the boundary of $Q$.  
(5) Edges intersect the boundary of $Q$ in a finite set of points, and do not use a corner of $Q$.
Conditions (1-5) can easily be achieved if arbitrary
curves are allowed for edges as follows:   (1) holds after adding sufficiently
many edges (which can be deleted in the final visibility representation), 
(2) can be achieved by re-routing the horizontal boundary of $Q$ along a so-called {\em tambourine}
\cite{CDF18}, (3) holds after a shear and (4-5) hold after locally re-routing.

Assume first that $G$ is toroidal, so $Q$ is a rectangular flat torus.
Enumerate the edges that intersect the
bottom side of $Q$ as $(s_i,t_i)$ (for $i=1,\dots,d$) from left to right,
named such that part
of the edge that goes upward from the bottom side ends at $s_i$ for $i=1,\dots,d$.
(This is feasible by condition (2) above.)
Create a new graph $G_{st}$ by removing edges $(s_i,t_i)$ for $i=1,\dots,s$,
adding a new vertex $t$ incident to $t_1,\dots,t_d$ and a new vertex $s$ 
incident to $s_1,\dots,s_d$.      
See Fig.~\ref{fig:K7_I}(a).
%(This may create multi-edges and faces bounded
%by two edges, but neither will bother us.)

Now assume that $G$ is embedded on a rectangular flat Klein bottle $Q$ instead. 
We construct $G_{st}$ in almost the same way, but the enumeration of edges is different.
%Namely, enumerate 
Let the edges that cross the bottom side of
$Q$ be $(s_1,t_d),\dots,(s_d,t_1)$ from left to right, 
named such that the part
of the edge that goes upward from the bottom side ends at $s_i$ for $i=1,\dots,d$.
Since the top and bottom sides of $Q$ are identified in opposite direction,
the order of edges along the top side of $Q$ is $(s_d,t_1),\dots,\allowbreak(s_1,t_d)$
from left to right.
Remove these edges and replace them by a vertex $s$ incident to $s_1,\dots,s_d$
and a vertex $t$ incident to $t_1,\dots,t_d$. 
See Fig.~\ref{fig:G_st}(a).

In both cases, by placing $t$ above the top side  of $Q$
and $s$ below the bottom side of $Q$, we obtain an embedding of $G_{st}$
on the flat cylinder, so it is a {\em plane} graph (i.e., drawn on the plane
with a fixed embedding).
The edges incident to $s$ lead to $s_1,\dots,s_d$ (in clockwise order)
and the edges incident to $t$ lead to $t_1,\dots,t_d$ (in counter-clockwise order).

\begin{observation}
Graph $G_{st}$ is 2-connected.
\end{observation}
\begin{proof}
Since $G_{st}$ is a plane graph, 2-connectivity is equivalent to all faces
being bounded by a simple cycle.     This holds for all faces of
$G$ by assumption.
The only faces of $G_{st}$ that are not in $G$ 
are those incident to $s$ and $t$. These consist of part of the boundary
of a face of $G$, plus two newly added edges that both end at $s$ (or both end
at $t$).   So the boundary of these faces are simple cycles as well.
\end{proof}

\begin{figure}[t]
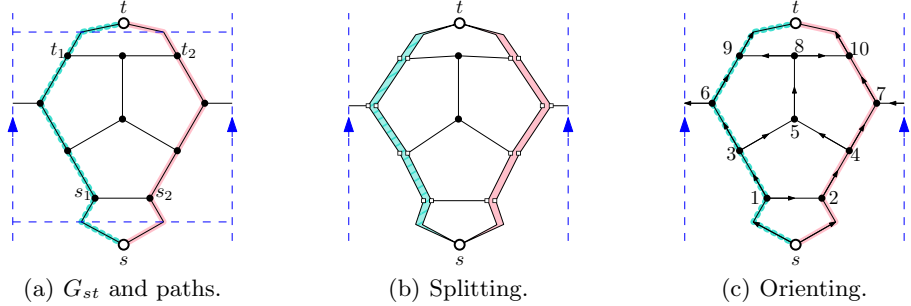

\subfigure[$G_{st}$ and paths.]{ \includegraphics[width=0.25\linewidth,page=3]{Petersen.pdf} }
\hspace*{\fill}
\subfigure[Splitting.]{ \includegraphics[width=0.25\linewidth,page=4]{Petersen.pdf} }
\hspace*{\fill}
\subfigure[Orienting.]{ \includegraphics[width=0.25\linewidth,page=5]{Petersen.pdf} }
\caption{The first few steps for the Petersen-graph from Fig.~\ref{fig:K7}(b).}
\label{fig:G_st}
\label{fig:splitPath}
\end{figure}

\iffalse
\begin{figure}[t]
\subfigure[~]{
\raisebox{5mm}{\includegraphics[width=0.22\linewidth,page=1]{K7.pdf}}
\hspace*{\fill}
\includegraphics[width=0.22\linewidth,page=3]{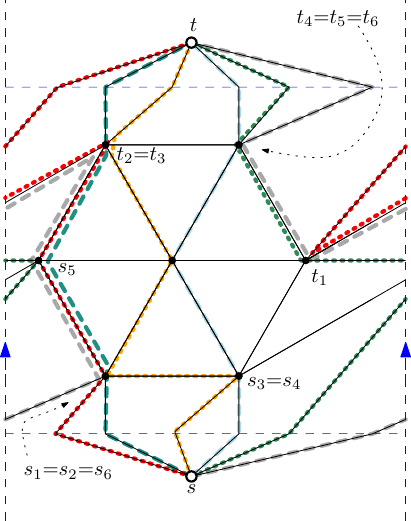}
}
\hspace*{\fill}
\subfigure[~]{
\raisebox{5mm}{\includegraphics[width=0.22\linewidth,page=1]{Petersen.pdf}}
\hspace*{\fill}
\includegraphics[width=0.22\linewidth,page=3]{Petersen.pdf}
}
\caption{Turning $G$ into $G_{st}$ for (a) $K_7$ and (b) the Petersen-graph.
We also show the path system (dashed/dotted).}
\label{fig:G_st}
\end{figure}
\fi

%----------------------------------------------------------------------
\subsection{Choosing Paths}
\label{subse:choosePaths}

We now show how to choose a set $\Pi$ of paths in $G_{st}$ that satisfy some properties.
A path is called {\em simple} if no vertex repeats.
Two simple edge-disjoint paths $\pi,\pi'$ are {\em non-crossing} if
at any vertex $v$ that is interior to both the paths only touch, i.e.,
the edges of the paths appear in order $\pi,\pi,\pi',\pi'$ in $\rho_v$.
%{\em cross
%at a vertex $v$} if $v$ is an interior vertex of both paths,
%and the edges of the paths at $v$ appear in order $\pi,\pi',\pi,\pi'$
%in $\rho_v$.

\begin{lemma}
\label{lem:findPaths}
\label{lem:choosePaths}
There exists a planar graph $\hat{G}$ (obtained by duplicating edges of $G_{st}$)
and a set of simple edge-disjoint non-crossing paths $\pi_1,\dots,\pi_d$ in $\hat{G}$
such that path $\pi_i$ begins with $(s,s_i)$ and ends with $(t_i,t)$ for $i=1,\dots,d$.
\end{lemma}

Before giving the proof, we need to define the operation of 
{\em splitting} a map at a path $\pi$
(also used in Figs.~\ref{fig:K7_I}(b) and \ref{fig:splitPath}(b)).
Temporarily direct $\pi$ from one end to the other.   
Duplicate all
interior vertices of $\pi$ (say vertex $v$ becomes $v^\ell$ and $v^r$) and
duplicate all edges of $\pi$ correspondingly.
%to connect the corresponding vertices.   
For any interior vertex $v$ of $\pi$, and any edge $e$ incident to $v$ but not on $\pi$,
we re-connect $e$ to end at $v^\ell$ [$v^r$] if
$e$ %emits `to the left' or `to the right', by which we mean whether $e$
occurs before [after] the outgoing edge of $\pi$ at $v$ when enumerating
$\rho_v$ beginning with the incoming edge of $\pi$ on $v$.
Splitting at $\pi$ creates a new face $f_{\pi}$ bounded by the two copies of $\pi$.
\begin{proof}
%(Sketch) Since we are allowed to duplicate edges we can make the
%minimum cut sufficiently big, and then find $d$ edge-disjoint paths from $s$ to $t$.
%Because the edge orders are compatible
%at $s$ and $t$, the paths can be chosen to have no crossing. A detailed proof is in the appendix.
Let $\pi$ be a simple path that begins with $(s,s_1)$ and ends with $(t_1,t)$;
this exists since $G_{st}$ is 2-connected.   Temporarily split graph $G_{st}$ at $\pi$ to obtain
a planar graph $\tilde{G}$.   The resulting new face $f_\pi$ contains both $s$ and $t$; for ease
of description we assume that $f_\pi$ is the outer-face of $\tilde{G}$.

Let $\tilde{G}^+$ be the graph obtained from $\tilde{G}$
by replacing any edge $e$ that is not incident to $s$ or $t$ by a multi-edge that has
$d+1$ copies of $e$. Any
$s$-$t$-cut of ${\tilde{G}}^+$ either consists of the edges incident to $s$
(then it has size $d+1$ since $(s,s_1)$ exists twice in $\tilde{G}$) or 
of the edges incident to $t$ (then it likewise has size $d+1$), or it contains some
edge $e$ not incident to either $s$ or $t$ and so 
has size at least $d+1$.  By the max-flow-min-cut
theorem  
%(see e.g.~\cite{AMO93}) 
therefore $\tilde{G}^+$ has a flow of
value $d{+}1$ from $s$ to $t$; equivalently, it has $d{+}1$ edge-disjoint
paths $\pi_1,\dots,\pi_{d+1}$ from $s$ to $t$.   Since $s$ and $t$ are both
on the outer-face we can find these paths
using {\em right-first search} \cite{RWW97}; this will automatically
make them crossing-free.

Since the paths are crossing-free and use all edges incident to $s,t$, and since $s$ and $t$ are 
on the outer-face, there is no choice which pair of edges must be the first and last
on each path.   The clockwise order of edges at $s$ (beginning after the outer-face)
is $(s,s^r_1),\dots,(s,s_d),(s,s^\ell_1)$.
The counter-clockwise order of edges at $t$ (beginning
after the outer-face) is $(t,t^r_1),\dots,(t,t_d),(t,t^\ell_1)$.   
Therefore path $\pi_i$ begins with $(s,s_i)$ and end with $(t_i,t)$ for $i=2,\dots,d$,
while $\pi_1$ and $\pi_{d+1}$ use the copies of $s_1$ and $t_1$.

To obtain $\hat{G}$, re-combine any two vertices $v^\ell$ and $v^r$ that resulted
from splitting an interior vertex $v$ of $\pi$, and keep all edges of $\tilde{G}^+$ except $(s,s_1^\ell)$
and $(t_1^\ell,t)$.
Since these two edges were used by $\pi_{d+1}$, they were used by no other path in $\pi_1,\dots,\pi_d$,
and we have hence obtained our desired path-system.
\end{proof}

%----------------------------------------------------------------------
\subsection{A Path-Constrained $st$-order}
\label{subse:stOrder}

By Lemma~\ref{lem:choosePaths}, we can fix a supergraph $\hat{G}$ of $G_{st}$
and a {\em path-system} $\Pi$, i.e.,
a set of simple edge-disjoint non-crossing paths from $s$ to $t$.
To draw $\hat{G}$, we add vertices one-by-one, and to draw the paths in $\Pi$
vertically, we require a vertex-order with special properties.  

We need some definitions.  A {\em bipolar orientation} 
is an assignment of directions to the edges that is acyclic and has exactly
one source and one sink.  An {\em $st$-order} is a vertex order $v_1,\dots,v_n$ 
such that orienting all edges from the lower-indexed to the higher-indexed
vertex gives a bipolar orientation.   Vice versa, for any bipolar orientation,
enumerating the vertices in topological order gives an $st$-order.   It is
well-known that any 2-connected graph has a bipolar orientation, even if we
fix a-priori which vertices should be the source and sink \cite{LEC67}; it
can be found in linear time \cite{ET76}.

We say that a bipolar orientation {\em respects a path system} $\Pi$ if
every path in $\Pi$ is directed from $s$ to $t$ in the bipolar orientation.
We phrase the following result for an arbitrary graph $H$ since it does
not depend on the graph stemming from a toroidal or Klein-bottle graph
and may be of independent interest.

\begin{lemma}
\label{lem:goodSTorder}
Let $H$ be a 2-connected plane graph with two vertices $s\neq t$. 
Let $\Pi$ be a set of simple edge-disjoint crossing-free 
paths from $s$ to $t$. Then $H$ has a bipolar orientation that respects $\Pi$
and has source $s$ and sink $t$.
\end{lemma}
\begin{proof}
%Enumerate the paths in $\Pi$ as $\pi_1,\dots,\pi_d$
%according to the clockwise order in which their first edges leave $s$, and note
%that the counter-clockwise order in which they arrive at $t$ must also be
%$\pi_1,\dots,\pi_d$ since paths do not cross.
Consider the graph $\hat{H}$ obtained from $H$ by splitting $H$ at each path in $\Pi$.   
%This gives for each path $\pi\in \Pi$ a face $f_\pi$ bounded by the two copies of $\pi$.
%Since $s,t\in f_\pi$, we can insert an edge connecting $s$
%to $t$ inside $f_\pi$.   
%(This creates multi-edges, but they will not bother us.)
%Call the resulting graph $\hat{H}$.   
See Figs.~\ref{fig:K7_I}(b) and \ref{fig:splitPath}(b).
Any face of $\hat{H}$ is either a face of $H$ (then it is a simple cycle
since $H$ is 2-connected)    or $f$ is bounded by the two copies of some
path $\pi\in \Pi$ (then it is a simple cycle since $\pi$ is simple).
So $\hat{H}$ is
2-connected and has a bipolar orientation $\hat{D}$ with source $s$ and
sink $t$.

It is well-known \cite{TT86} that in $\hat{D}$
any face has a unique source and sink.   
%(Most previous drawing-applications
%used this in the special case where $s$ and $t$ are on a common face, but
%one verifies from the proof in \cite{TT86} that this holds even if they
%are not.)
In any face $f_\pi$ bounded by two copies of some $\pi\in \Pi$, 
the unique source is $s$ and the unique sink is $t$.   Therefore both copies
of $\pi$ are directed from $s$ to $t$ and undoing the splitting gives
the desired orientation.
\end{proof}

%----------------------------------------------------------------------
\subsection{Path-Constrained Visibility Representations}
\label{subse:VR}

In this section, we give an easy construction of a 
visibility representation on the flat cylinder
where a given path-system $\Pi$ is drawn vertically.
Formally, we say that a path $\pi$ lies {\em on an exclusive column} $\ell$ (in 
a visibility representation $\Gamma$)
if all edges of $\pi$ are represented by segments on $\ell$, and column 
$\ell$ intersects no vertex- or edge-segment except the ones that belong 
to vertices/edges of $\pi$.

Our approach to create visibility representations
is quite different from prior constructions 
%\cite{OW78,TT86,RT86,Wis85,TT91,MR98,Hut05}.   
\cite{Hut05,MR98,OW78,RT86,TT86,TT91,Wis85},
which either read the coordinates for the segments
directly from the orientation (using the length of the
longest paths in the primal and dual graph), or reduced the graph
(or its dual) by removing an edge somewhere in the graph, creating 
a representation recursively, and expanding.   In contrast to this,
we use here an incremental approach which resembles more the
incremental approaches taken for straight-line drawings \cite{CDF18,FPP90} 
or orthogonal drawings \cite{BK98}. This uses a vertex
ordering and adds the vertices to the drawing one-by-one.

\begin{theorem}
\label{thm:cylinder}
\label{thm:cyl}
Let $H$ be a 2-connected plane graph with two vertices $s,t$ 
and let $\Pi$ be a set of simple edge-disjoint non-crossing paths from $s$ to $t$.
Then $H$ has a visibility representation on the flat cylinder such that
each $\pi\in \Pi$ lies on an exclusive column.
\end{theorem}
\begin{proof}
Fix a bipolar orientation using Lemma~\ref{lem:goodSTorder} and 
extract an $st$-order $v_1,\dots,v_n$ from it; we know $v_1=s$ and $v_n=t$
and the numbers along any path in $\Pi$ increase from $s$ to $t$.
For $i=1,\dots,n$ let $H_i$ be the subgraph induced by $v_1,\dots,v_i$ and
let the {\em cut} $E_{i{:}i{+}1}$ be the set of all edges $(v_h,v_j)$ with
$h\leq i<j$.     There is a natural cyclic order of the edges in $E_{i{:}i{+}1}$
implied by the embedding of $H$ (specifically, if we contracted
the vertices $v_1,\dots,v_i$ into a supernode, then the order of $E_{i{:}i{+}1}$ 
would be the clockwise order of edges at this supernode).
We will use induction on $i$ to create a visibility representation of $H_i$
on a flat cylinder that satisfies the following for $i<n$:
\begin{enumerate}
\item 
Every edge $e=(v_h,v_j)$, $h{<}j$ in cut $E_{i{:}i{+}1}$ is associated with a column
that intersects $v_h$ and that is empty above $v_h$.
\item 
The left-to-right order of columns associated with $E_{i{:}i{+}1}$ respects the cyclic
order of edges in $E_{i{:}i{+}1}$.
\item 
For any path $\pi\in \Pi$, the sub-path of $\pi$ in $H_i$ lies on an
exclusive column, and the same column is associated with the unique
edge of $\pi$ in $E_{i{:}i{+}1}$.
\end{enumerate}

Fig.~\ref{fig:VR}(a-b) illustrates the following construction.
For $i=1$, we create the desired visibility representation simply by
defining a horizontal line segment $s(v_1)$ for $v_1$ with $y$-coordinate 0
and width $|E_{1{:}2}|$, and assigning columns intersecting $s(v_1)$
to edges in $E_{1{:}2}$ in the correct order.   

For $i>1$, assume we have created a visibility representation of $H_{i-1}$
already.  Define edge-sets $E_i^-:=\{(v_h,v_i): h<i\}$ and
$E_i^+:=\{(v_i,v_j): i<j\}$; the former is non-empty by $i>1$ since we have an $st$-order.
%, i.e., the incident edges of $v_i$ split by
%whether the other endpoint is earlier or later in the $st$-order.   
It is well-known \cite{LEC67} that $E_i^-$ is consecutive in
the cyclic order of edges in $E_{i{-}1{:}i}$.   By the invariant therefore 
there exists an $x$-interval  $X_i$ on the flat cylinder that intersects all columns
associated with edges in $E^-_i$ in its interior and intersects no other columns
associated with $E_{i{-}1{:}i}$.    Define the segment $s(v_i)$ of $v_i$ to 
have $x$-range $X_i$ and a $y$-coordinate that is higher than the one of all its
neighbours in $E_i^-$. These edges can then be completed along
their associated columns. 

%It remains 
To associate columns with $E_i^+$, we insert new
columns as needed.
%(Formally, we can
%{\em insert} a new column at non-integral $x$-coordinate $x'$ by
%increasing by 1 the $x$-coordinate of every edge-segment and 
%vertex-segment-endpoint whose current $x$-coordinate is bigger than $x'$.)
First consider any edge $e\in E_i^+$ in some path $\pi\in \Pi$.
Since $\pi$ begins at $s$ and $i>1$, and since indices increase along $\pi$, 
some edge $e'\in E^-_i$ also belongs to $\pi$.   Associate the column of $e'$ with $e$.
Notice that this associates columns in the correct order, because
if multiple paths $\pi_1,\dots,\pi_k\in \Pi$ all went through $v_i$,
then the counter-clockwise order of their edges in $E_i^-$ at $v_i$
must be the same as the clockwise order of their edges in $E_i^+$ at $v_i$,
otherwise two of these paths would cross at $v_i$.

Now consider any edge $e\in E_i^+$ that does not belong to a path in $\Pi$.
Assign a ray upward
from $s(v_i)$ to $e$, choosing rays such that all edges in $E_i^+$ use
distinct rays/columns and their order reflects the order of edges
at $v_i$.   
%This is feasible since all intersections of $s(v_i)$ with associate columns lie in its interior.   
By stretching horizontal
segments as needed, we can re-assign coordinates so that all inserted rays
lie on integer coordinates, hence become new columns.
This gives the desired visibility representation of $H_i$.
\end{proof}

\begin{figure}[ht]
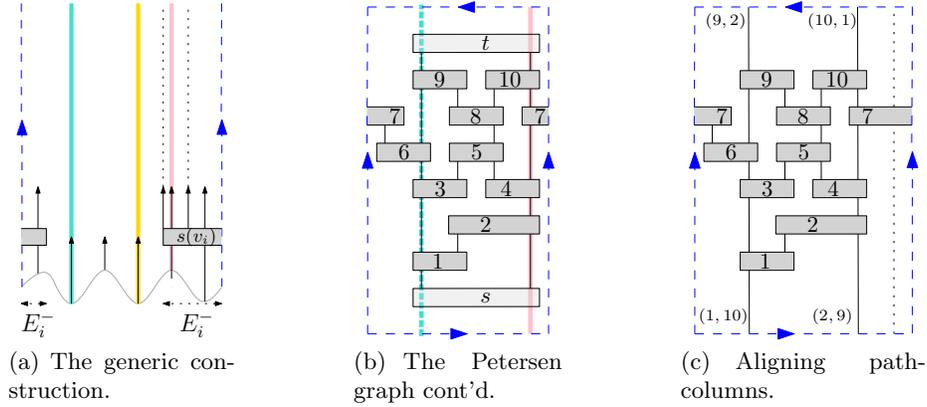

\subfigure[The generic construction.]{
\includegraphics[height=45mm,page=8]{Petersen.pdf}
}
%\hspace*{\fill}
%\subfigure[~]{
%\includegraphics[height=40mm,page=6]{K7.pdf}
%}
\hspace*{\fill}
\subfigure[The Petersen graph cont'd.]{
\includegraphics[height=45mm,page=6]{Petersen.pdf}
}
\hspace*{\fill}
\subfigure[Aligning path-columns.]{
\includegraphics[height=45mm,page=7]{Petersen.pdf}
}
\caption{Creating visibility representations. } 
\label{fig:VR}
\label{fig:K7VR}
\end{figure}

\subsection{Putting It All Together}

We now have all ingredients to prove our main result (Theorem~\ref{thm:tor}): Any toroidal
graph $G$ without loops has a visibility representation on the rectangular flat torus.
See Fig.~\ref{fig:K7_I} for the entire process.

\begin{proof}
Add edges to $G$ until all its faces are simple cycles.
As described in Subsections~\ref{subse:makePlanar}-\ref{subse:VR},
split $G$ at edges $(s_i,t_i)$ (for $i=1,\dots,d$) to obtain $G_{st}$,
find a supergraph $\hat{G}$ with a path-system $\Pi$ where path $\pi_i$ begins with $(s,s_i)$ and
ends with $(t_i,t)$, find an orientation  that respects $\Pi$, and find a visibility representation $\Gamma$
of $\hat{G}$ on the flat cylinder $Q$ such that $\pi_i$ is
drawn along an exclusive column $\ell_i$.   Remove the segments that
represent $s$ and $t$ and 
%extend $(s_i,s)$ and $(t_i,t)$ along $\ell_i$
%until they reach the horizontal boundary of $Q$.   
complete $(s_i,t_i)$ along column $\ell_i$.
After re-interpreting
$Q$ as a rectangular flat torus this gives the desired visibility
representation of $G$ after deleting all added edges.
%See also Fig.~\ref{fig:K7_I}.
\end{proof}

\iffalse

\def\myScale{0.68}
\begin{figure}[ht]
\hspace*{\fill}
\begin{tabular}{ccc}
{\includegraphics[scale=\myScale,page=3,trim=0 10 0 0,clip]{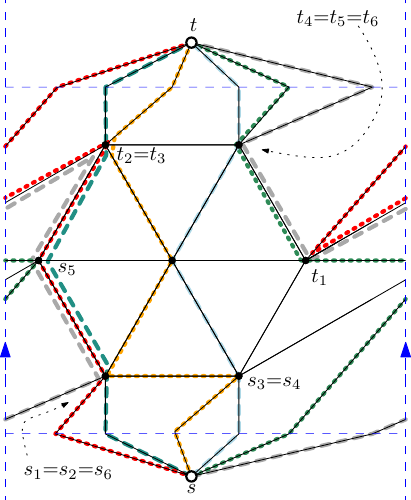}}
& ~~~~~
& {\includegraphics[scale=\myScale,page=4,trim=0 10 0 0,clip]{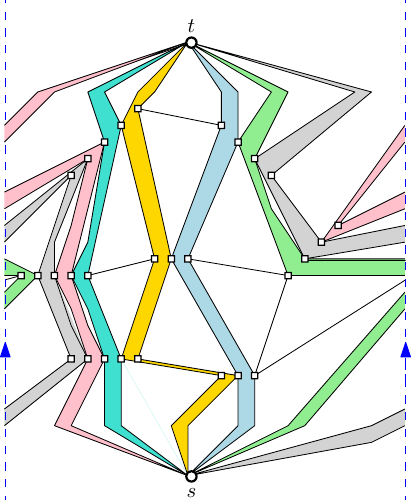}}  \\
(a) The paths. && (b) Splitting.\\[2ex]
{\includegraphics[scale=\myScale,page=5,trim=0 10 0 0,clip]{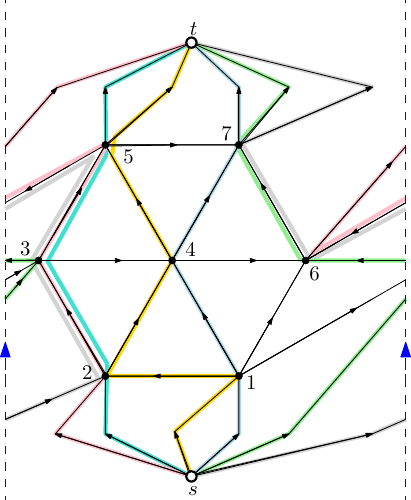}}
& 
&{\includegraphics[scale=0.8,page=6,trim=0 10 0 0,clip]{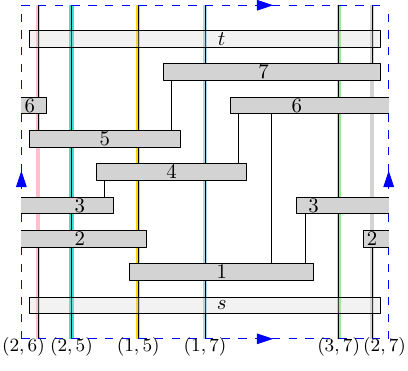}}\\
(c) The orientation. && (d) The final drawing.
\end{tabular}
\hspace*{\fill}
\caption{The construction for the complete graph $K_7$.}
\label{fig:K7_I}
\end{figure}
\fi

With a bit more care when reconnecting edges, 
the same approach also works for Klein-bottle graphs.

\begin{theorem}
\label{thm:Klein}
Let $G$ be a graph without loops embedded on the Klein bottle.   Then $G$ has a visibility representation
on the rectangular flat Klein bottle.
\end{theorem}
\begin{proof}
Exactly as in the previous proof, create a visibility representation $\Gamma$
of $\hat{G}$ on the flat cylinder $Q$ such that $\pi_i$ is
drawn along an exclusive column $\ell_i$.   Remove the segments that
represent $s$ and $t$ and extend $(s_i,s)$ and $(t_i,t)$ along $\ell_i$
until they reach the horizontal boundary of $Q$.   

We are not quite done yet, because we must ensure that column $\ell_i$ 
`lines up' with column $\ell_{d+1-i}$ (for $i=1,\dots,\lfloor d/2 \rfloor$)
so that edges $(s_i,t_{d+1-i})$ and $(s_{d+1-1},t_i)$ are connected
correctly when interpreting $Q$ as the flat Klein bottle.   This is
easily achieved by inserting columns.  Namely, assume $Q$ has
$x$-range $[0,w]$ and let $x(\ell)$ denote the $x$-coordinate of 
column $\ell$.   For $i=1,\dots,\lfloor d/2\rfloor$, while 
$x(\ell_i)<w{-}x(\ell_{d+1-i})$, insert an empty column to the left of $\ell_i$,
and while $x(\ell_i)>w-x(\ell_{d+1-i})$, insert an empty column to the right of $\ell_{d+1-i}$ 
See Fig.~\ref{fig:VR}(c).
This maintains distances of $\ell_1,\dots,\ell_{i-1}$ to the left boundary
and distances of $\ell_{d+2-i},\dots,\ell_d$ to the right boundary of $Q$.
So performing this for $i=1,\dots,\lfloor d/2 \rfloor$ gives the desired
visibility representation on the flat Klein bottle.
\end{proof}

We note here that our visibility representations exactly respect the given
embedding.   Under this restriction, the condition `no loops' cannot be avoided.   
(This was essentially observed by Mohar and Rosenstiehl \cite{MR98} already.)
Namely, let $M_0$ be a graph with a single vertex 
$v$ and two loops $\ell_1,\ell_2$ such that $\rho_v=\langle \ell_1,\ell_2,\ell_1,\ell_2\rangle$.
This is toroidal, but has no visibility representation on the rectangular flat torus that
respects the embedding since the rotation scheme at $v$ in such an embedding is
necessarily $\ell_1,\ell_1,\ell_2,\ell_2$.

\begin{figure}[ht]
\hspace*{\fill}
\subfigure[~]{\includegraphics[scale=1,page=4]{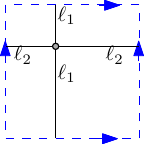}}
\hspace*{\fill}
\subfigure[~]{\includegraphics[scale=1,page=5]{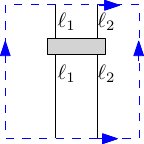}}
\hspace*{\fill}
\subfigure[~]{\includegraphics[scale=1,page=6]{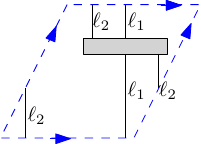}}
\hspace*{\fill}
\caption{(a) Graph $M_0$.   (b) The only possible visibility representation on a rectangular flat torus.   (c) An embedding-preserving visibility on the flat torus.} 
\label{fig:M0}
\end{figure}

%%%%%%%%%%%%%%%%%%%%%%%%%%%%%%%%%%%%%%%%%%%%%%%%%%%%%%%%%%%%%%%%%%%%%%%%
%\section{Further remarks}
%\label{sec:remarks}

\subsection{Grid-Size}

We can give a bound on the grid-size of Theorem~\ref{thm:tor},
%the resulting visibility representation, 
assuming that the input is already a map (i.e., all faces 
are simple cycles).
We say that a visibility representation {\em has
grid-size $w\times h$} 
if the fundamental rectangle $Q$ intersects $w$ columns and $h$ rows,
not counting the boundaries of $Q$.
%if the fundamental rectangle $Q$ is $[\tfrac{1}{2},w+\tfrac{1}{2}]\times [\tfrac{1}{2},h+\tfrac{1}{2}]$
%(i.e., we assume nothing is placed exactly on the boundaries of $Q$, and they do not count for the width/height).
%on which it
%is defined overlaps $w$ distinct integral $x$-coordinates and $h$ distinct integral
%$y$-coordinates.   (So on a rectangular flat torus, rectangle $Q$ would
%be a $[0,w]\times [0,h]$-rectangle since $x$-coordinate $h$ coincides with
%$x$-coordinate 0.)
In our current approach, the visibility representation
$\Gamma_{st}$ of $G_{st}$ uses significantly more area than it needs to 
since we may duplicate quite a few edges of $G_{st}$ to obtain the path
system (see also the discussion below).   However, as
for all visibility representations, one should apply compaction steps
(similar as for VLSI design \cite{Len90}) to reduce the size of the
drawing.   We claim that after doing this, 
the visibility representation $\Gamma$ of a toroidal graph $G$ has grid-size 
at most $(m-n)\times n$.

To see this, observe that we need at most $n$ rows, 
since assigning row $i$ to vertex $v_i$ will certainly place it high enough
and the rows for $s,t$ can be deleted during compaction.
As for the number of columns, each column must contain at least one edge, else
it could have been deleted.
Furthermore, we used a bipolar orientation of $\hat{G}$, which means
that every vertex other than $s$ and $t$ has both an incoming and an outgoing
edge.   Since $\hat{G}$ is obtained from $G_{st}$ by duplicating edges, the
same holds in $G_{st}$.
Vertices $s$ and $t$ are removed in the final visibility representation
(but their incident edges remain and are re-combined).   With the standard
compaction steps, therefore at least one column at each vertex $v$ is used for
two edges incident to $v$.   It follows that each vertex saves at least one column,
hence the number of columns is $m-n$.

\subsection{Run-time}

Following the steps of our algorithm, it is very clear that our visibility
representations can be found in polynomial time.   
%(Finding a non-contractible
%cycle can be done in polynomial time; in fact, we can even find an embedding
%on the flat torus \cite{LPVV01}.)   
In fact, the drawing in Theorem~\ref{thm:cyl} can be found in linear
time with standard-approaches:
%in linear time in the size of the graph , with similar tools as in \cite{CDF18,BK98}: 
do not explicitly maintain the $x$-coordinates, but store the drawing implicitly
by computing $x$-spans of vertex-segments and $x$-offsets of edge-segments
from the left endpoints of their lower endpoints.    The final drawing can
then be computed with one pass over the entire graph after all vertices have
been placed.

Unfortunately finding the drawings in Theorems~\ref{thm:tor} and \ref{thm:Klein}
may take superlinear time since the supergraph $\hat{G}$ may have many extra 
edges.  If $G_{st}$ has $\Omega(n)$
disjoint edge-cuts that separate $s$ and $t$, then each of the $|\Pi|$
paths must duplicate an edge in each edge-cut, leading to $\Omega(m+|\Pi|n)$ 
edges for $\hat{G}$.  One can show that $|\Pi|\in O(\sqrt{n})$
can be achieved, because any toroidal graph has a non-contractible cycle
of length $O(\sqrt{n})$ \cite{AH78}, and we can use such a cycle in the dual
graph to find an embedding where $O(\sqrt{n})$ edges cross the horizontal side
and hence necessitate a path in $\Pi$.  With this choice we get
$|\hat{G}|\in O(n^{1.5})$ and run-time $O(n^{1.5})$.
%, and parsing it to compute the visibility
%representation appears the time-bottleneck. 

Reducing this to linear time seems not implausible: we 
need the paths in $\Pi$ only to steer us towards placing edges in
the visibility representation at a suitable place, and it may be
possible to encode this in a smaller data structure that permits
linear run-time.  This remains for future work.

\section{Other Drawing Styles} 

We close the paper by discussing how our results do (or do not)
imply results in some other graph drawing styles that are closely related to visibility
representations.  The first drawing style that we consider are
{\em orthogonal point-drawings}, where vertices
are represented by points and every edge is a polygonal curve between
its endpoints that uses only horizontal and vertical segments and 
does not intersect other edges or vertices.      (These can only
exist if the graph has maximum degree at most four.)

\begin{theorem}
Every toroidal graph with maximum degree four has an orthogonal
point-drawing on the rectangular flat torus.  
Every Klein-bottle graph with maximum degree four has an orthogonal
point-drawing on the flat Klein bottle.  
\end{theorem}
\begin{proof}
Tamassia and Tollis \cite{TT89} showed how to create orthogonal
point-drawings by starting with a visibility representation and
replacing vertex-segments locally by points and polygonal
curves that connect to the edge-segments.    The exact same
transformations can be applied to any visibility representation that
lies on a flat representation, so using it with
Theorem~\ref{thm:tor} and Theorem~\ref{thm:Klein} (after subdividing
loops, if any) gives 
the desired orthogonal point-drawings.
\end{proof}

Two other related drawing styles are grid contact and 
tessellation representations.   %Recall that
A {\em bipartite graph} has a vertex-partition $V=W\cup B$ such
that there are no edges within $W$ or within $B$.   
In a {\em grid contact representation} of a bipartite graph, 
the vertices of $W$ and $B$ are assigned to horizontal
and vertical segments, respectively, with all segments disjoint 
except that any segment of one kind may touch at both of its ends an
interior point of a segment of the other kind, and such a common
point occurs only if the two vertices are adjacent.   
See Fig.~\ref{fig:K44}(b).
It is well-known \cite{FMP91} that
% could also cite CKU98
every planar bipartite graph has a grid contact representation in the plane,
and Mohar and Rosenstiehl \cite{MR98} showed that any toroidal bipartite graph has a grid
contact representation on the flat (not necessarily rectangular) torus. 
A {\em tessellation representation} of a graph $G$ is a grid contact
representation of the 
%{\em face-vertex-incidence graph}, i.e., the
bipartite
graph whose vertices are the faces and vertices of $G$ and whose edges
are the incidences between them.%
\footnote{In contrast to earlier work \cite{MR98}, we use here {\em weak} models, where not all adjacencies that
could be added must exist.}   
See Fig.~\ref{fig:K44}(c).

\begin{figure}[ht]
\hspace*{\fill}
\subfigure[~]{\includegraphics[scale=1,page=3]{K44.pdf}}
\hspace*{\fill}
\subfigure[~]{\includegraphics[scale=1,page=1]{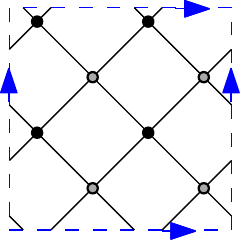}}
\hspace*{\fill}
\subfigure[~]{\includegraphics[scale=1,page=2]{K44.pdf}}
\hspace*{\fill}
\caption{(a) A set of segments that is a grid contact representation of $K_{4,4}$ (shown in (b)) or
a tessellation representation of the graph in (c).}
\label{fig:K44}
\end{figure}

Mohar and Rosenstiehl constructed tessellation representations of toroidal
graphs (on a flat torus), from which their results on grid contact representations
and visibility representations follow easily.   They must permit a non-rectangular
flat torus because they reduce their graph to $M_0$ (or another single-vertex graph with loops),
which cannot be represented on a rectangular flat torus.
%They also state that 
%their graphs $M_0$ and $M_1$ (
%the graph consisting of a single vertex with two or three non-contractible
%loops attached has no tessellation representation on the
%rectangular flat torus.   
But does it help to have no loops?

\begin{conjecture}
\label{conj:tess}
Every toroidal graph without loops has a tessellation
representation on the rectangular flat torus.
\end{conjecture}
\begin{conjecture}
Every bipartite toroidal graph without loops has a grid contact representation
on the rectangular flat torus.
\end{conjecture}

At first sight one might think that Theorem~\ref{thm:tor} implies
Conjecture~\ref{conj:tess}, because Mohar and Rosenstiehl \cite{MR98} show that 
a visibility representation can be converted to a tessellation representation.
Alas, their definition of  ``visibility
representation'' uses the `strong' model where {\em all}
visibilities must lead to an edge, hence faces are triangles,
and this is vital in their proof.
On the positive side, their proof does not affect the shape of
the flat representation, so using it one can show that
Conjecture~\ref{conj:tess} holds for
toroidal graphs where all faces are triangles.

%
%that it would be very easy to modify the algorithm
%in Section~\ref{subse:VR} to create vertical segments for every 
%face of $G_{st}$ that touches the horizontal segments of its
%incident vertices (after lengthening those).   The difficulty is any
%face $f$ that is incident to an edge $(s_i,t_i)$ of $G$.   Such
%a face would be split into two faces $f',f''$ when converting $G$ to $G_{st}$,
%and it is not clear whether we can force an alignment of the vertical
%segments inserted for $f',f''$ .

Finally we are interested in %representing toroidal graphs as
{\em segment intersection representations}, i.e., every vertex is assigned
to a segment (of arbitrary slope) on the flat torus, with
segments intersecting if and only if the vertices are adjacent.
%One interesting open question concerns representing toroidal graphs 
%(and perhaps also graphs embedded on the Klein bottle)
%as segment-intersection graphs. 
%(This is what motivated our interest in grid contact representations
%of bipartite toroidal graphs, which is a first step.)
%Specifically, given a toroidal graph, can we assign a segment on the flat torus to
%every vertex such that two segments intersect if and only if the
%two vertices are adjacent?      
Such representations exist
for all planar graphs \cite{CG09}, and one proof of this proceeds by
representing a planar graph as the intersection-graph
of L-shaped curves in the plane \cite{GIP17} and then converting the
L-shaped curves into segments \cite{MP92}.
%, see also \cite{Bie-IPL19}).    
The corresponding questions on the flat torus
appear to be open even if we drop `rectangular': 
\begin{question}
Does every simple toroidal graph have a segment intersection representation
on the flat torus?
\end{question}
\begin{question}
Is every simple toroidal graph the intersection-graph of L-shaped curves
on the flat torus?
\end{question}
\begin{question}
If a graph is the intersection-graph of L-shaped curves
on the flat torus, then is it also the intersection-graph
of segments on the flat torus?
\end{question}

Finally all these questions could be asked also for graphs
embedded on the Klein bottle (or other surfaces, such as the
projective plane).

\iffalse
%%%%%%%%%%%%%%%%%%%%%%%%%%%%%%%%%%%%%%%%%%%%%%%%%%%%%%%%%%%%%%%%%%%%%%%%
\section{Conclusion}
\label{sec:concl}

In this paper, we studied visibility representations of toroidal 
graphs and Klein-bottle graphs on the rectangular flat torus/Klein bottle.   
We show that this always exists, with the crucial insight that we
can create visibility representations of planar graphs such that
a set of edge-disjoint non-crossing paths is kept on vertical lines (and use these to 
reconnect edges that we had to split to turn our input graph into
a planar graph). 
Our result also leads to grid contact representations
(for bipartite graphs) and orthogonal point-drawings (for
graphs with maximum degree 4) of graphs embedded on the
torus/Klein bottle.
\fi

%%%%%%%%%%%%%%%%%%%%%%%%%%%%%%%%%%%%%%%%%%%%%%%%%%%%%%%%%%%%%%%%%%%%%%%%
\bibliographystyle{splncs04}
\bibliography{journal,full,gd,papers}

\end{document}